\documentclass[11pt]{article}

\usepackage[english]{babel}

\usepackage[letterpaper,top=2cm,bottom=2cm,left=3cm,right=3cm,marginparwidth=1.75cm]{geometry}
\usepackage[utf8]{inputenc}

\usepackage{fullpage}
\usepackage{hyperref,verbatim, enumerate, amsmath,amssymb,amsfonts,mathrsfs, amsthm, mathtools, bbm, algorithm, algpseudocode}
\usepackage[capitalize]{cleveref}
\usepackage{comment}
\usepackage{fullpage}
\usepackage{xcolor}
\usepackage{lipsum}
\usepackage[font=itshape]{quoting}
\usepackage[bottom]{footmisc}
\usepackage{csquotes}
\usepackage{todonotes}
\usepackage[normalem]{ulem}

\newcommand{\iso}{\text{unique}}

\newcommand{\res}{\text{residual}}
\newcommand{\alt}{\text{alt}}
\newcommand{\inp}{\text{input}}
\newcommand{\cata}{\text{catalytic}}
\newcommand{\emphdef}[1]{{\sf {#1}}}

\newcommand\restr[2]{{
  \left.\kern-\nulldelimiterspace 
  #1 
  \vphantom{\big|} 
  \right|_{#2} 
  }}
\newcommand{\poly}{\mathrm{poly}}
\allowdisplaybreaks[4]

\usepackage{blindtext}
\usepackage{enumitem}
\usepackage{caption}
\usepackage{placeins}
\usepackage{graphicx}

\usepackage{fancybox}
\usepackage{hyperref}

\usepackage{amsmath,amsfonts,amsthm,amssymb,xcolor}
\usepackage{bm} 
\usepackage{bbm}

\usepackage{nicefrac}

\usepackage{float}

\newtheorem{theorem}{Theorem}[section]
\newtheorem{corollary}[theorem]{Corollary}
\newtheorem{lemma}[theorem]{Lemma}

\newtheorem{claim}[theorem]{Claim}
\newtheorem{fact}[theorem]{Fact}

\newtheorem*{theorem*}{Theorem}
\newtheorem*{corollary*}{Corollary}
\newtheorem*{conjecture*}{Conjecture}
\newtheorem*{lemma*}{Lemma}
\newtheorem*{thm*}{Theorem}
\newtheorem*{prop*}{Proposition}
\newtheorem*{obs*}{Observation}

\newtheorem*{remark*}{Remark}
\newtheorem*{rec*}{Recommendation}

\newtheorem{definition}[theorem]{Definition}
\newtheorem{remark}[theorem]{Remark}
\newtheorem*{definition*}{Definition}

\newenvironment{fminipage}%
  {\begin{Sbox}\begin{minipage}}%
  {\end{minipage}\end{Sbox}\fbox{\TheSbox}}

\DeclareMathOperator{\polylog}{polylog}

\newcommand{\newclass}[2]{\newcommand{#1}{{\text{\upshape\sffamily #2}}\xspace}}
\renewcommand{\P}{{\text{\upshape\sffamily P}}\xspace}
\newclass{\NP}{NP}
\newclass{\ZPP}{ZPP}
\newclass{\coNP}{coNP}
\newclass{\BPP}{BPP}
\newclass{\Logspace}{L}
\newclass{\NL}{NL}
\newclass{\coNL}{coNL}
\newclass{\UL}{UL}
\newclass{\coUL}{coUL}
\newclass{\BPL}{BPL}
\newclass{\PL}{PL}
\newclass{\prBPL}{prBPL}
\newclass{\PSPACE}{PSPACE}
\newclass{\EXP}{EXP}
\newclass{\EXPSPACE}{EXPSPACE}
\newclass{\TIME}{TIME}
\newclass{\SPACE}{SPACE}
\newclass{\NSPACE}{NSPACE}
\newclass{\SC}{SC}
\newclass{\coNSPACE}{coNSPACE}
\newclass{\BPSPACE}{BPSPACE}
\newclass{\TFNP}{TFNP}

\newclass{\NC}{NC}
\newclass{\NCo}{NC$^1$}
\newclass{\ACz}{AC$^0$}
\newclass{\ACo}{AC$^1$}
\newclass{\SACo}{SAC$^1$}
\newclass{\TC}{TC}
\newclass{\TCz}{TC$^0$}
\newclass{\TCo}{TC$^1$}
\newclass{\NCt}{NC$^2$}
\newclass{\RNC}{RNC}
\newclass{\PSDNC}{pseudo-deterministic NC}
\newclass{\NUSPL}{non-uniform SPL}
\newclass{\RNCt}{RNC$^2$}
\newclass{\RNCtt}{RNC$^3$}
\newclass{\RNCo}{RNC$^1$}
\newclass{\QNC}{Quasi-NC}
\newclass{\VP}{VP}
\newclass{\CL}{CL}
\newclass{\CLP}{CLP}
\newclass{\CSPACE}{CSPACE}

\newcommand{\nearSC}{\textsf{TIME-SPACE(}\poly(n), n^{1 - \Omega(1)}\textsf{)}}
\newclass{\GapL}{GapL}

\newclass{\MATCH}{MATCH}
\newclass{\DET}{DET}
\newclass{\LOSSY}{LOSSY}
\newclass{\LOSSYNC}{LOSSY[NC]}
\newclass{\LOSSYC}{LOSSY[$\mathcal{C}$]}
\newcommand{\Cclass}{\mathcal{C}}
\newclass{\ZPC}{ZP-$\mathcal{C}$}
\newclass{\ZPNC}{ZPNC}
\newclass{\Comp}{Comp}
\newclass{\Decomp}{Decomp}

\newif\ifblind

\blindtrue

\title{Bipartite Matching is in Catalytic Logspace}

\author{}
\ifblind
\author{Aryan Agarwala \\ Max-Planck-Institut f\"{u}r Informatik \\ \texttt{aryan@agarwalas.in} \and Ian Mertz\thanks{Partially supported by grant 24-10306S of GA \v{C}R and supported by Center for Foundations of Contemporary Computer Science
(Charles Univ. project UNCE 24/SCI/008).} \\ Charles University \\ \texttt{iwmertz@iuuk.mff.cuni.cz}}
\fi
\date{}

\begin{document}

\setlength{\abovedisplayskip}{5pt}
\setlength{\belowdisplayskip}{5pt}

\pagenumbering{gobble}

\maketitle

\begin{abstract}

\noindent
Matching is a central problem in theoretical computer science,
with a large body of work spanning the last five decades.
However, understanding matching in the time-space bounded setting
remains a longstanding open question, even in the presence of additional
resources such as randomness or non-determinism.\\

\noindent
In this work we study space-bounded machines with access to catalytic space,
which is additional working memory that is full with arbitrary data that
must be preserved at the end of its computation. Despite this heavy restriction,
many recent works have shown the power of catalytic space,
its utility in designing classical space-bounded algorithms,
and surprising connections between catalytic computation and derandomization.\\

\noindent
Our main result is that bipartite maximum matching ($\MATCH$)
can be computed in catalytic logspace ($\CL$) with a polynomial time bound ($\CLP$). Moreover, we show that $\MATCH$ can be reduced to the lossy coding problem for
$\NC$ circuits ($\LOSSYNC$).
This has consequences for matching, catalytic space, and derandomization:

\begin{itemize}
   \item \textbf{Matching}: this is the \textit{first} well studied subclass of $\P$ which is known to compute $\MATCH$, as well as the \textit{first}
   algorithm simultaneously using sublinear free space and polynomial time with
   \textit{any} additional resources. Thus, it gives a potential path
   to designing stronger space and time-space bounded algorithms.
   \item \textbf{Catalytic space}: this is the \textit{first} new problem
   shown to be in $\CL$ since the model was defined, and one which is
   \textit{extremely} central and well-studied. Furthermore, it implies a strong barrier
   to showing $\CL$ lies \textit{anywhere} in the $\NC$ hierarchy, and suggests
   to the contrary that $\CL$ is even more powerful than previously believed.
   \item \textbf{Derandomization}: we give the \textit{first} class $\Cclass$
   beyond $\Logspace$ for which we exhibit a natural problem in $\LOSSYC$ which
   is not known to be in $\Cclass$, as well as a \textit{full derandomization
   of the isolation lemma} in $\CL$ in the context of $\MATCH$. This
   also suggests a possible approach to derandomizing the famed $\RNC$ algorithm
   for $\MATCH$.
\end{itemize}
Our proof combines a number of strengthened ideas from isolation-based algorithms
for matching alongside the compress-or-random framework in catalytic computation.
\end{abstract}


\pagebreak

\pagenumbering{arabic}

\section{Introduction}
\label{sec:intro}

In this work we study a number of key questions and models of complexity
between logarithmic space ($\Logspace$) and polynomial time ($\P$).
In particular we focus on the relationship between
\textit{bipartite maximum matching} ($\MATCH$) and
\textit{poly-time bounded catalytic logspace} ($\CLP$),
as well as their implications for efficient parallel algorithms ($\NC$)
and efficient time-space algorithms ($\nearSC$).
Lastly we draw on connections between both $\MATCH$ and $\CLP$ to problems
in derandomization---for the former we focus on the \textit{isolation lemma},
and with regards to the latter we discuss reductions to the \textit{lossy coding
problem}---to make progress therein.

\subsection{Matching and catalytic computation}

\paragraph{Matching.}
In the $\MATCH$ problem, we are given a bipartite graph $G$ as input
and our goal is to return a subset of edges of maximum size such that no two edges
share an endpoint. $\MATCH$ has been a central problem in the study of
complexity since its inception, and was one of the earliest problems
to be studied with respect to time; it has been known for 70 years
that $\MATCH$ can be solved in $\P$ \cite{Kuhn55,HopcroftKarp}.
However, we are not aware of any well-studied class\footnote{Throughout this paper,
when we discuss classes containing $\MATCH$ we ignore granular poly-time
classes which immediately follow as a direct consequence of the above algorithms,
e.g. $\TIME[n^2]$, or those that contain matching by definition, i.e. the
class of problems reducible to $\MATCH$.}
$\Cclass \subseteq \P$ such that $\MATCH \in \Cclass$.
\begin{center}\label{question:match_in_subp}
    \textbf{Question 1:} Is $\MATCH$ in any subclass of $\P$?
\end{center}

\noindent
For two such classes $\Cclass$ in particular,
namely $\NC$ and $\SC$, proving this would be a major breakthrough~\cite{Lovasz79, KarpUpfalWigderson85, MulmuleyVaziraniVazirani87, AllenderReinhardtZhou99, MahajanVaradarajan00, DattaKulkarniRoy10, FennerGurjarThierauf16, SvenssonTarnawski17, AnariVazirani19, GoldwasserGrossman17, MillerNaor89, AnariVazirani20, Barnes-etal}.
With regards to parallel complexity, a long line of work culminated
in showing that $\MATCH$ can be solved in randomized
$\NC$ ($\RNC$)~\cite{MulmuleyVaziraniVazirani87} and in quasi-polynomial size
$\NC$ ($\QNC$)~\cite{FennerGurjarThierauf16}, but despite decades of research
no true $\NC$ algorithms are known to this day.

\begin{center}\label{question:match_in_nc}
    \textbf{Question 2:} Is $\MATCH \in \NC$?
\end{center}

\noindent
Matching holds an even more central place in the study of space
and time-space efficiency.
It is widely conjectured that logarithmic space is \textit{insufficient} to solve $\MATCH$
(i.e. $\MATCH \notin \Logspace$), and proving so would give a breakthrough separation
between $\Logspace$ and $\P$.
Similarly, with regards to time-space complexity we have no
algorithms which compute $\MATCH$ in polynomial time and sublinear space
even given additional resources such as non-determinism or randomness,
and it is unclear whether such algorithms should exist or not. 

\begin{center}\label{question:match_in_sc}
    \textbf{Question 3:} Do there exist any resources $\mathcal{B}$ such that

    \noindent
    $\MATCH \subseteq \mathcal{B}\nearSC$?
\end{center}

\noindent
In this work we study another such class called \textit{catalytic logspace} ($\CL$),
and in particular its poly-time variant $\CLP$, which is of great interest
in relation to both $\NC$ and $\SC$. In doing so we give a first-ever
solution to Questions 1 and 3, as well as a potential
barrier---or approach---to resolving Question 2.


\paragraph{Catalytic computing.}
In catalytic computing, a space-bounded machine is given additional access to a much longer ``catalytic'' tape, which is additional memory already full of arbitrary data whose contents must be preserved by the computation. $\CL$ is the class of problems solvable by a logspace machine augmented with a polynomial length catalytic tape, and $\CLP$ is the subclass of $\CL$ where the machine is additionally required to run in polynomial time.\\

\noindent
The framework of catalytic computation was formally introduced by Buhrman et al. \cite{BuhrmanCleveKouckyLoffSpeelman14} in order to understand the power of additional but used memory. It was informally conjectured earlier, in the context of the \textit{tree evaluation problem}~\cite{CookMckenzieWehrBravermanSanthanam12}, that used memory could not grant additional power to space bounded machines. However, \cite{BuhrmanCleveKouckyLoffSpeelman14} showed that $\CL$ and even $\CLP$ contain problems believed to not be in $\Logspace$:
$$\Logspace \subseteq \NL \subseteq \TCo \subseteq \CLP \subseteq \CL \subseteq \ZPP$$

\noindent
In the same work, they state that it is unclear what their result implies about the strength of $\CL$. Particularly, \textit{is $\Logspace \subsetneq \CL$, or is the intuition that catalytic space does not grant additional power indeed correct, thus giving an approach to proving $\Logspace = \NL = \TCo = \CL$?}

\begin{center}\label{question:cl_vs_l}
    \textbf{Question 4:} Where does $\CL$ lie between $\Logspace$ and $\ZPP$?
\end{center}


\noindent
Since their result, many works have studied the utility and power of
catalytic computation (see e.g.~\cite{BuhrmanKouckyLoffSpeelman18,GuptaJainSharmaTewari19,DattaGuptaJainSharmaTewari20,BisoyiDineshSarma22,CookMertz22,Pyne24,CookLiMertzPyne25,GuptaJainSharmaTewari24,FolkertsmaMertzSpeelmanTupker25,PyneSheffieldWang25,KouckyMertzPyneSami25}).
These results and techniques have also seen applications for ordinary
space-bounded computation, such as 1) work on the \textit{Tree Evaluation Problem}
by Cook and Mertz~\cite{CookMertz20, CookMertz21, CookMertz24} as well as a
subsequent breakthrough by Williams~\cite{Williams25} for simulating time in low space; and
2) win-win arguments for derandomization and other questions in 
logspace~\cite{DoronPyneTell24, LiPyneTell24, Pyne24} such as a recent result
of Doron et al. \cite{DoronPyneTellWilliams25} showing that, in an instance-wise
fashion, either $\NL \subseteq \SC$ or $\BPL \subseteq \NL$. We refer
the interested reader to surveys of Kouck\'{y}~\cite{Koucky16} and Mertz~\cite{Mertz23}
for more discussion. \\


\noindent
However, despite all of the aforementioned work, \textit{no problems outside of $\TCo$ have
been shown to be in $\CL$}\footnote{Li, Pyne, and Tell~\cite{LiPyneTell24} give a
\textit{search} problem in $\CL$ which is not known to be in $\TCo$; however,
unlike with $\MATCH$, the corresponding decision problem is in $\TCo$ and thus
is not applicable for studying the power of $\CL$ with regards to decision problems.}.
Thus, it is still possible that $\CL \subseteq \TCo$, which has led to conjectures,
such as that of \cite{Koucky16,Mertz23}, that $\CL$ can indeed be computed
\textit{somewhere} in the $\NC$ hierarchy.

\begin{center}\label{question:cl_vs_nc}
    \textbf{Question 5:} Is $\CL \subseteq \NC$?
\end{center}

\paragraph{Our results (1).}
In this work we address and connect all of the aforementioned questions by
showing that catalytic logspace can compute bipartite matching in polynomial time:

\begin{theorem}
\label{thm:main}
    $$\MATCH \in \CLP$$
\end{theorem}

\noindent
We make a few notes on the power of Theorem~\ref{thm:main}:

\begin{enumerate}[leftmargin=*]
    \item This is the first subclass of $\P$ which has been shown to contain $\MATCH$ (\textbf{Question 1}).
    
    \item This is the first $\mathcal{B}\nearSC$ algorithm for $\MATCH$ for any additional resources $\mathcal{B}$ (\textbf{Question 3}). We believe that this gives hope for showing that $\MATCH$ can be computed in $\nearSC$, or perhaps, as was the case
    with Tree Evaluation, such catalytic algorithms are a reason to doubt our intuition
    that $\MATCH$ cannot be solved in $o(\log^2 n)$ space altogether.

    \item This is the first problem outside $\TCo$ shown to be in $\CL$, and thus the first strengthening of $\CL$ since the original work of Buhrman et al.~\cite{BuhrmanCleveKouckyLoffSpeelman14};
    this also gives the strongest evidence thus far that $\Logspace \neq \CL$
    (\textbf{Question 4}).
    
    \item This ties together \textbf{Question 2 \& 5}, as showing
    $\CL$ or even $\CLP$ is contained in $\NC$ would immediately give the
    breakthrough result $\MATCH \in \NC$.
    Conversely, if one believes that $\MATCH \notin \NC$, then this gives
    evidence to contradict the conjectures of \cite{BuhrmanCleveKouckyLoffSpeelman14,Koucky16,Mertz23} that $\CL \subseteq \NC$.

    \item Many previous works have shown reductions from other graph problems
    to $\MATCH$, and thus our $\CLP$ inclusion extends to these problems
    as well. We discuss the cases of \textit{weighted matching}, $s$-$t$ \textit{max-flow},
    and \textit{global min-cut} in Section~\ref{sec:others}, although we stress this list
    is only a partial sample of such extensions of Theorem~\ref{thm:main}.
\end{enumerate}

\begin{figure}[t]
\label{fig-cl}
    \centering
    \includegraphics[width=1\linewidth]{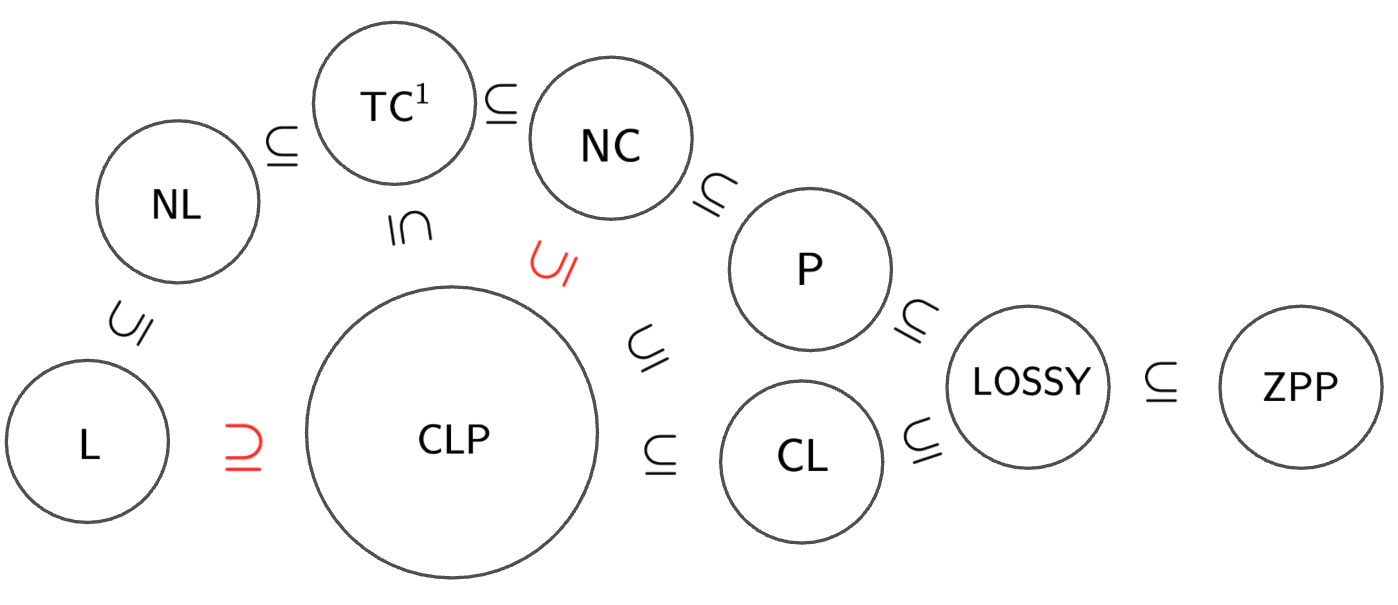}
    \caption{Barriers: black inclusions have been proven before this work, and red inclusions were conjectured in \cite{BuhrmanCleveKouckyLoffSpeelman14, Koucky16, Mertz23}; our inclusion of $\MATCH$ in $\CLP$ acts as a barrier to all of those conjectures.}
\end{figure}


\subsection{Derandomization}
\paragraph{Lossy coding.}
One other important, and useful, aspect of catalytic computation is a
connection to fundamental questions in derandomization,
a connection exploited in recent line of work~\cite{DoronPyneTell24, LiPyneTell24, Pyne24, DoronPyneTellWilliams25}
for showing novel non-catalytic space-bounded algorithms.
The \emphdef{lossy coding} problem~\cite{Korten22} is defined as follows: given a pair of
circuits $\Comp:\{0, 1\}^n \rightarrow \{0, 1\}^{n-1}$ and
$\Decomp : \{0, 1\}^{n-1} \rightarrow \{0, 1\}^n$, output
$x \in \{0, 1\}^n$ such that $\Decomp(\Comp(x)) \neq x$.
For a closed complexity class $\Cclass$, $\LOSSYC$ is the class of problems
$\Cclass$-reducible to lossy coding when $\Comp$ and $\Decomp$ are required to
be computed and evaluated in the class $\Cclass$. It is easy to see that 
$$\Cclass \subseteq \LOSSYC \subseteq \ZPC $$
\noindent
$\textsf{LOSSY[P]}$, which is often just referred to as $\LOSSY$,
has seen attention in recent years
in the context of total function $\NP$~\cite{Korten22}, range avoidance~\cite{KortenPitassi24},
and meta-complexity~\cite{ChenLuOliveiraRenSanthanam23}, and was discussed at
length in a recent survey by Korten~\cite{Korten25} on the topic.
The problem of derandomizing $\LOSSY$ is seen as a stepping stone towards
the major derandomization goal of $\P = \ZPP$.\\

\noindent
Unfortunately, for most classes $\Cclass$ it remains unclear whether
$\LOSSYC$ contains \textit{any natural and well studied problem} outside
$\Cclass$.
This was posed as an open problem in \cite{Korten25}.

\begin{center}\label{question:c-vs-lossyc}
    \textbf{Question 6:} Does $\LOSSY[\Cclass]$ admit a natural problem not known to be in $\Cclass$ for any class $\Cclass$?
\end{center}

\noindent
The only known progress on this question comes via catalytic computation,
as Pyne~\cite{Pyne24} showed that $\BPL \subseteq \textsf{LOSSY[L]}$. However, this remains unknown for all larger classes. Moreover, due to the fact that space-bounded randomized classes use read-once randomness, \textsf{LOSSY[L]} is perhaps an overkill for derandomizing $\BPL$.

\paragraph{Isolation lemma.}
One of the greatest contributions of the study of $\MATCH$ to complexity theory is
a key tool, known as the \textit{isolation lemma}, which is the backbone of
all parallel algorithms for $\MATCH$ algorithms since its introduction by
Mulmuley, Vazirani, and Vazirani~\cite{MulmuleyVaziraniVazirani87}. It
has since turned out to be a very strong tool used in the design of
randomized algorithms for a wide range of problems ~\cite{OrlinStein93,LingasPersson15,NarayananSaranVazirani94,GurjarThierauf17,KlivansSpielman01,AllenderReinhardt00,BourkeTewariVinodchandran09,KalampallyTewari16,MelkebeekPrakriya19, ArvindMukhopadhyay08} \footnote{For a more comprehensive list of
applications of the isolation lemma we direct the readers to \cite{AgarwalGurjarThierauf20}.}.
Thus, derandomizing the isolation lemma, independent of a concrete problem,
has become an important problem in itself~\cite{ChariRohatiSrinivasan93,ArvindMukhopadhyay08,AgarwalGurjarThierauf20,AnariVazirani20,GurjarThieraufVishnoi21}.

\begin{center}\label{question:derand-iso}
    \textbf{Question 7:} Which deterministic classes $\Cclass$ can implement
    the isolation lemma?
\end{center}

\noindent
While the breakthrough $\QNC$ algorithm of Fenner, Gurjar, and
Thierauf~\cite{FennerGurjarThierauf16} does exactly this, it both lies
outside $\P$ and uses weights which are much larger than in all other applications.
In fact this question is unresolved even for the ``derandomization'' class
for $\NC$ discussed above, namely its $\LOSSY$ variant.

\begin{center}\label{question:iso-lossy}
    \textbf{Question 8:} Is the isolation lemma in $\LOSSYNC$?
\end{center}

\paragraph{Our results (2).}
An analysis of our algorithm in Theorem~\ref{thm:main} gives answers to all
of the above questions:

\begin{theorem} \label{thm:lossy}
    $$\MATCH \in \LOSSYNC$$
\end{theorem}

\noindent
Again some discussion is in order:
\begin{enumerate} [leftmargin=*]
    \item To the best of our knowledge, this is the first well studied and natural
    problem shown to be in $\LOSSYC$ and not known to be in $\Cclass$
    for any class $\Cclass$ larger than $\Logspace$ (\textbf{Question 6}).
    
    \item Since $\NC \subseteq \LOSSYNC \subseteq \ZPNC \subseteq \RNC$, this is a
    direct improvement over $\MATCH \in \RNC$ \cite{Lovasz79, KarpUpfalWigderson85, MulmuleyVaziraniVazirani87} and $\MATCH \in \ZPNC$ \cite{Karloff86} (it is incomparable to all other $\NC$-related results about $\MATCH$). It also motivates the further study of \emphdef{lossy coding}, since any derandomization of $\LOSSYNC$ would make progress towards the ultimate goal of $\MATCH \in \NC$. The partial derandomization of $\BPL$ by Doron et al. \cite{DoronPyneTellWilliams25} using $\textsf{LOSSY[L]}$ gives hope that this approach may prove fruitful.

    \item Our algorithm is built upon a catalytic derandomization of the
    isolation lemma, thus showing that it can be
    implemented in deterministic $\CLP$ (\textbf{Question 7}) as well as
    $\LOSSYNC$ (\textbf{Question 8}).
    We believe that the framework introduced in this work could be used to show similar $\CL$ and $\LOSSY$ results for the many other problems which admit isolation lemma based algorithms - similar to the exciting line of work which followed \cite{FennerGurjarThierauf16} (see e.g. \cite{SvenssonTarnawski17,GurjarThierauf17,GurjarThieraufVishnoi21,GoldwasserGrossman17,AnariVazirani19,AnariVazirani20,KalampallyTewari16,MelkebeekPrakriya19}).
\end{enumerate}

\begin{figure}
    \centering
    \includegraphics[width=1\linewidth]{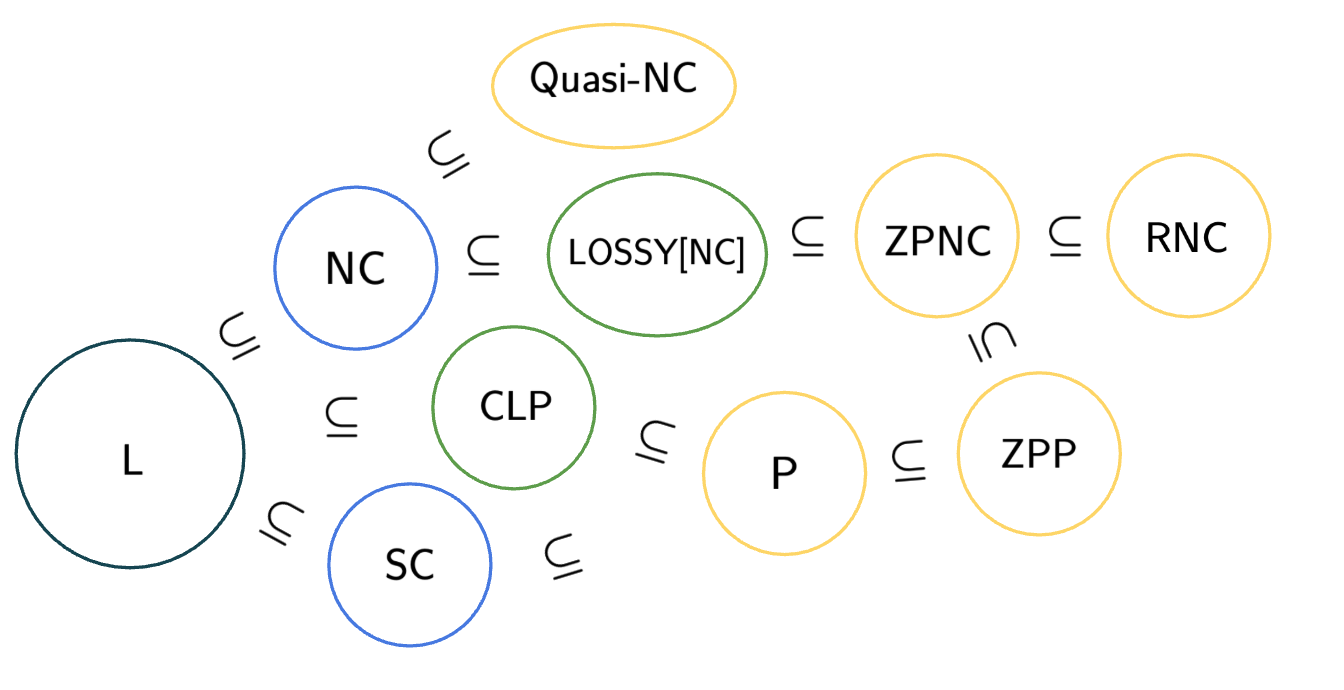}
    \caption{Contributions: yellow bubbles correspond to classes which were already known to contain $\MATCH$, blue bubbles correspond to classes for which the $\MATCH$ inclusion is a long standing open problem, and green bubbles correspond to the inclusions shown in this paper. }
\end{figure}


\subsection{Open problems}
We suggest a number of open problems coming from this work:
\begin{enumerate}
    \item Can $\MATCH$ be computed using less catalytic space, say
    $O(\log^2 n)$, or alternatively can our algorithm be used as a subroutine
    for ordinary space or time/space-bounded computation?

    \item Can $\CL$ prove related functions such as non-bipartite matching,
    matroid intersection, etc., or indeed can the inclusion $\TCo \subseteq \CL$ be
    improved to a stronger class, perhaps even $\NCt$?
    Alternatively, can we use the $\CL$ derandomization
    of the isolation lemma to show novel derandomizations such as
    $\RNCo \subseteq \CL$?
    \item Does $\LOSSYNC$ contain any other related functions such as non-bipartite matching or matroid intersection? Is $\NC = \LOSSYNC$?
\end{enumerate}

\subsection{Proof overview}

We finish the introduction by giving an overview of our argument
along with where in the paper each part will appear.
Our algorithm will utilize the \textit{isolation lemma} based framework of Mulmuley, Vazirani,
and Vazirani~\cite{MulmuleyVaziraniVazirani87} for \MATCH, combined with the
\textit{compress-or-random} framework introduced by Cook et al.~\cite{CookLiMertzPyne25},
with novel insights for both.

\paragraph{Weighting and Isolating Matchings (Section~\ref{sec:find-match}).}
Let $n$ be the number of vertices in the graph.
Assume we have access to some fixed weight function $W$ whose weights are at most
$\poly(n)$, and consider the set $\mathcal{M}$ of all perfect matchings $M$ in $G$.
We say that $W$ is \textit{isolating} on $\mathcal{M}$ if there
exists a single matching $M_{\iso} \in \mathcal{M}$ of minimum weight.
\cite{MulmuleyVaziraniVazirani87} prove that a random $W$ is isolating on
$\mathcal{M}$ with high probability. \\

\noindent
By taking the Edmonds matrix $N$ of $G$,
where $N[i,j] = 2^{w(i,j)}$ for an isolating weight assignment $W$,
it is additionally proven in ~\cite{MulmuleyVaziraniVazirani87}
that the unique min-weight perfect matching $M_{\iso}$
can be derived from the low-order bit of $\DET(N)$, and thus
is in $\CLP$ given access to $W$. We describe modifications that can find a matching of any fixed size $k \in [n]$ as well, again provided that such
a matching is the unique minimum weight matching of size $k$.

\paragraph{Maximum Matching via Isolation (Section~\ref{sec:induction}).}
Let $k$ be a value for which $W$ isolates a matching of size $k$,
and let $M^k_{\iso}$ be the matching in question.
There are three possibilities:  $M^k_{\iso}$ is the largest matching in $G$, there is a unique min-weight matching of size $k+1$ under $W$,
or there are at least two min-weight matchings of size $k+1$ under $W$. To distinguish
these cases, we will build a directed graph $G^k$ based on $G$ and $M^k_{\iso}$ which
contains two additional nodes, $s$ and $t$, such that, from every path $P$ from $s$
to $t$, we can extract a matching $M_P$ of size $k+1$ along with its weight\footnote{For readers familiar with matching theory, $G^k$ is exactly the residual graph of $G$ with respect to $M^k$, and the path $P$ is an augmenting path.}.\\

\noindent
Thus if there are no $s$-$t$ paths in $G^k$ we can output $k$; here
we use the fact that $\NL \subseteq \CLP$ to perform this search. Our main contribution in this section is on the structure of the graph $G^k$, which allows us to test whether a matching of size $k+1$ is isolated by $W$, and more importantly, handle the case where it is not isolated.

\paragraph{Compression of Weight Assignments (Section~\ref{sec:compress}).}
The main observation in the current paper is in identifying and handling the remaining case,
where there exist at least two min-weight matchings of size $k+1$, in $\CLP$.
Notice that for a random $W$ this case will not occur with high probability
due to \cite{MulmuleyVaziraniVazirani87},
but we have not yet specified where $W$ comes from in our deterministic
procedure. In fact, we will have $W$ come from
the (adversarial) catalytic tape itself,
and we will handle this case in the style of \cite{CookLiMertzPyne25}. \\

\noindent
To recap, matchings of size $k+1$ can be found by checking for paths
in $G^k$, and in particular we can find an edge $e$ in $G$
that is in \textit{some} min-weight matchings of size $k+1$ \textit{but not all} of them.
Our key is to notice that the weight
of $e$ is in fact redundant; by constructing $G^k$, and from it $M_1^{k+1}$ and
$M_2^{k+1}$ - which are min-weight matchings of size $k+1$ not containing and containing $e$ respectively, we find that $w(e) = w(M_1^{k+1}) - w(M_2^{k+1} \smallsetminus \{e\})$. Thus since
$W$ is written on the catalytic tape, we
can erase $w(e)$ from the catalytic tape and free up $\Theta(\log n)$ bits.
We will then iterate this procedure, using a new part of the catalytic
tape as a replacement for the erased memory.

\paragraph{Final Algorithm (Section~\ref{sec:final-alg}).}
To sum up, we begin by splitting the catalytic tape into two parts,
one part for computation and another part for specifying
a weight assignment $W$ on the edges $E$ of $G$, where each
weight has some $\Theta(\log n)$ bits, plus $\poly(n)$ reserve weights
which we save for later. \\

\noindent
Starting from $k = 1$ we certify that $W$ isolates a min-weight matching of size $k$,
which allows us to compute this matching $M_{\iso}^k$ in $\CLP$. We then
use $M_{\iso}^k$ to certify this for $k+1$, and if so then we proceed.
If not, then either there are no matchings of size $k+1$, at
which point we return $k$ and halt, or we find a redundant weight
in $W$. \\

\noindent
In the latter case, we erase this section of the catalytic tape,
replace it with a reserve weight that we had set aside, and start the algorithm over.
Repeating this $\poly(n)$ times, we either eventually isolate the largest
min-weight matching in the graph or we free up enough space to compute the
matching by brute force. At the end of the algorithm we recover the
compressed section of the tape one by one to reset the catalytic tape.



\section{Preliminaries}
\label{sec:prelims}

We use notation $\mathbb{Z}_{\leq c}$ to denote the non-negative integers
of value at most $c$. The \emphdef{determinant} function is denoted by
$\DET$.

\subsection{Graphs, matching, and weights}
We denote by $G = (V, E)$ a graph on vertices $V$ and edges $E$.
For $v \in V$, we define
\(E(v) \vcentcolon = \{e \in E \ | \ e \text{ is incident upon }v\}\).


\begin{definition}[Matching]
Let $G = (V, E)$ be an undirected graph.
A \emphdef{matching} is a set $M \subseteq E$ of edges such that
\[\forall v \in V, \ |E(v) \cap M| \leq 1 \]
\noindent
A vertex $v \in V$ is matched by $M$ (and otherwise unmatched) if 
\[|E(v) \cap M| = 1\]
\noindent
The \emphdef{size} of the matching $M$ is $|M|$, and we call a matching
\emphdef{perfect} if every vertex is matched.
\end{definition}

\noindent
In this paper we study matching in bipartite graphs $G = (V = L \sqcup R, E)$,
i.e. where all edges exclusively go between $L$ and $R$.
We focus on the case where $|L| = |R| = |V|/2$, and for convenience we define
$n := |V|/2$ rather than the size or number of nodes of $G$.
    
\begin{definition}
    The \emphdef{matching function}, denoted by $\MATCH$, takes
    as input an $n \times n$ bipartite graph $G$ and returns a matching $M$ in $G$ of maximum size. 
\end{definition}

\noindent
We will also work with different sorts of graphs for our algorithm and analysis.
First, we will sometimes work with directed graphs,
and will make it clear from context whether $G$ is directed or not.
We will also consider weighted graphs, i.e. graphs where
we are given a weight assignment $W:E \rightarrow \mathbb{Z}$;
for any $S \subseteq E$ we extend the definition of $W$ and define
$W(S) := \sum_{s \in S} W(s)$.

\begin{definition}[Symmetric Difference]
    Let $G = (V, E)$ be a graph.
    The \emphdef{symmetric difference} of matchings $M_1, M_2 \subseteq E$,
    denoted by $M_1 \Delta M_2$,
    is defined as having edges $(M_1 \smallsetminus M_2) \cup (M_2 \smallsetminus M_1)$.

    \noindent
    Similarly for any matching $M$ and set $S \subseteq E$, we define
    $M \oplus S$ as $(M \cup S) \smallsetminus (M \cap S)$.
\end{definition}

\noindent
Note that $M \oplus S$ need not be a matching, depending on $S$.

\subsection{Catalytic computation}

Our main computational model in this paper is the catalytic space model:
\begin{definition}[Catalytic machines]
    Let $s := s(n)$ and $c := c(n)$. A \emphdef{catalytic Turing machine} with
    space $s$ and \emphdef{catalytic space} $c$ is a Turing machine $M$ with a
    read-only input tape of length $n$, a write-only output tape, a read-write work tape of length $s$,
    and a second read-write work tape of length $c$ called the \emphdef{catalytic tape},
    which will be initialized to an adversarial string $\tau$. \\

    \noindent
    We say that $M$ computes a function $f$ if for every $x \in \{0,1\}^n$ and
    $\tau \in \{0,1\}^c$, the result of executing $M$ on input $x$ with initial
    catalytic tape $\tau$ fulfils two properties: 1) $M$ halts with $f(x)$ written on the output tape; and
    2) $M$ halts with the catalytic tape in state $\tau$.
\end{definition}

\noindent
Such machines naturally give rise to complexity classes of interest:

\begin{definition}[Catalytic classes]
    We define $\CSPACE[s,c]$ to be the family of functions computable by
    catalytic Turing machines with space $s$ and catalytic space $c$.
    We also define \emphdef{catalytic logspace} as
    $$\CL := \bigcup_{d \in \mathbb{N}} \CSPACE[d \log n, n^d]$$
    Furthermore we define $\CLP$ as
    the family of functions computable by $\CL$ machines that are
    additionally restricted to run in polynomial time for every
    initial catalytic tape $\tau$.
\end{definition}

\noindent
Important to this work will be the fact, due to Buhrman et al.~\cite{BuhrmanCleveKouckyLoffSpeelman14},
that $\CLP$ can simulate log-depth threshold circuits:

\begin{theorem}[\cite{BuhrmanCleveKouckyLoffSpeelman14}]
\label{thm:bckls}
    $$\TCo \subseteq \CLP$$
\end{theorem}

\noindent
This gives a number of problems, such as determinant and
s-t connectivity, in $\CLP$. We will mention a few of
these directly, as they will be necessary later.
First, determinant over matrices with polynomially many
bits is in $\GapL$ (see c.f.~\cite{MahajanVinay97})
and thus in $\CLP$:
\begin{lemma}
\label{lem:DET-in-cl}
    Let $N$ be an $n \times n$ matrix over a field of size $\exp(n)$.
    Then $\DET(N)$ is computable in $\CLP$.
\end{lemma}

\noindent
Second, we extend the inclusion of $\NL$ to show that that
\textit{weighted} connectivity is also in $\CLP$:
\begin{lemma}
\label{lem:compute-min-weight-path}
    Let $G = (V, E)$ be a directed graph, let $s,t \in V$ be a given source
    and target vertex, and $W: E \rightarrow \mathbb{Z}_{\leq \poly(n)}$ be edge
    weights such that $G$ does not have any non-positive weight cycles under $W$.
    Then there exists a $\CLP$ machine which, given $(G, s, t, W)$,
    computes the minimum weight of a simple $s$-$t$ path.
\end{lemma}
\begin{proof}
    The problem of deciding whether $G$ contains an $s$-$t$ walk of length
    $\leq |E|$ and weight $\leq \alpha$, for $\alpha \leq \poly(n)$,
    is trivially decidable in $\NL$ and thus in $\CLP$.
    Binary searching over the value of $\alpha$ gives the minimum $\alpha^*$ such that
    $G$ contains an $s$-$t$ walk of length $\leq |E|$ and weight $\leq \alpha^*$.
    Since all cycles in $G$ have strictly positive weight,
    the minimising walk is also guaranteed to be a simple $s$-$t$ path.
\end{proof}

\subsection{Other complexity classes}
Besides catalytic computation, we also work with a few other classes.
First we recall the classic $\NC$ definition of parallel complexity:

\begin{definition}[$\NC$]
    A language $\mathcal{L}$ is computable in (uniform) $\NC$ if there exists
    a (uniform) circuit family $\{C_n\}_{n \in \mathbb{N}}$ such that $C_n$
    has size $\poly(n)$, depth $\log^{O(1)} n$, and decides membership
    in $\mathcal{L}$ on all inputs of size $n$.
\end{definition}

\noindent
We also define the class of problems reducible to the \emphdef{lossy coding}
problem over various objects:

\begin{definition}[Lossy coding and $\LOSSYC$]
    The \emphdef{lossy coding problem} is defined as follows: given a pair of
    algorithms $\Comp:\{0, 1\}^n \rightarrow \{0, 1\}^{n-1}$ and
    $\Decomp : \{0, 1\}^{n-1} \rightarrow \{0, 1\}^n$ as input, our goal is
    to output any $x \in \{0, 1\}^n$ such that $\Decomp(\Comp(x)) \neq x$.\\

    \noindent
    Let $\mathcal{C}$ be a complexity class. We define $\LOSSYC$ as the
    set of languages reducible to the lossy coding problem whose input
    algorithms come from the class $\Cclass$. When no class is given,
    we define $\LOSSY := \LOSSY[\P]$.
\end{definition}

\noindent
Of note, the above subroutines (weighted $s$-$t$ connectivity and $\DET$)
are also in $\NC$, which will be useful for Theorem~\ref{thm:lossy}.

\section{Weighting and Isolating Matchings}
\label{sec:find-match}

Our first task is to set up the basic structure of our algorithm,
which is to find matchings via \textit{isolation}.

\begin{definition}
    Let $U$ be a set and $W:U \rightarrow \mathbb{Z}$ be a weight assignment,
    and let $\mathcal{F} \subseteq {2^U}$ be a family of subsets of $U$.
    We say that $W$ is \emphdef{isolating} for $\mathcal{F}$ if
    \[\exists! \ S_{\iso} \in \mathcal{F} \text{ such that } W(S_{\iso}) = \min_{F \in \mathcal{F}} W(F)\]
    where $\exists!$ means that exactly one such set exists.
\end{definition}

\noindent
Mulmuley, Vazirani, and Vazirani~\cite{MulmuleyVaziraniVazirani87}
showed that, given the ability to compute the determinant,
isolation is sufficient to solve perfect matching.

\begin{theorem}[\cite{MulmuleyVaziraniVazirani87}]
\label{thm:isolated-pm}
    Let $G = (V = L \cup R, E)$ be a bipartite graph, and let
    $W:E \rightarrow \mathbb{Z}_{\leq \poly(n)}$ be a weight assignment
    that isolates a perfect matching $M_{\iso}$ in $G$.
    Then there exists an $\Logspace^{\DET}$ machine which, given input $(G, e \in E)$ and
    access to $W$, determines whether $e \in M_{\iso}$.
\end{theorem}

\noindent
Thus Lemma~\ref{lem:DET-in-cl} gives us a $\CLP$ algorithm for
finding the isolated perfect matching of min-weight, as
Theorem~\ref{thm:isolated-pm} computes $\DET$ on an $n \times n$
matrix whose entries have $\poly(n)$ bits. \\

\noindent
We will also need to extend Theorem~\ref{thm:isolated-pm} to work for
matchings of any fixed size $k \in [n]$. We sum this up with the
observations above into our main algorithm for this section.

\begin{lemma}
\label{lem:isolated-km}
    Let $G = (V = L \cup R, E)$ be a bipartite graph, and let $k \in \mathbb{N}$
    be a parameter whose value is at most the size of the maximum matching of $G$.
    Let $W:E \rightarrow \mathbb{Z}_{\leq \poly(n)}$ be edge weights such that
    $W$ isolates a size $k$ matching $M^k_{\iso}$ in $G$.
    Then there exists a $\CLP$ machine which, given input $(G, e \in E)$ and
    access to $W$, determines whether $e \in M^k_{\iso}$.
\end{lemma}
\begin{proof}
    Given $G$, we construct the following graph $G' = (V', E')$ and weight function $W'$:
    \begin{enumerate}
        \item $V' = L' \cup R'$, where $L'$ consists of $L$ and $n-k$ new vertices
        and $R'$ is constructed similarly. 
        \item $E'$ consists of $E$ along with a clique between the new vertices
        on each side with the old vertices on the other side.
        That is, $E' = E \cup \{(u, v) \ | \ u \in L, v \in R' \smallsetminus R\} \cup \{(u, v) \ | \ u \in R, v \in L' \smallsetminus L\}$
        \item Index all vertices in $G'$ by $[2(n+n-k)]$ arbitrarily.
        For $e \in E$, we define $W'(e) = W(e) \cdot 10 \cdot n^4$,
        and $e = (u, v) \in E' \smallsetminus E$ we define $W'(e) = u \cdot v$ to be the
        the product of the indices of the vertices.
    \end{enumerate}

    \noindent
    By construction, every perfect matching $M' \subseteq E'$ in $G'$ corresponds
    to a size $k$ matching $M \subseteq E$ in $G$.
    For any size $k$ matching $M$ of $G$ and perfect matching $M'$ of $G'$,
    $M'$ is defined to be an extension of $M$ if $M' \cap E = M$.
    Note that for any perfect matchings $M_1$ and $M_2$ of $G'$,
    if $W(M_1 \cap E) < W(M_2 \cap E)$, then

    \begin{align*}
        W'(M_2) - W'(M_1) & = (W'(M_2 \cap E) + W'(M_2 \smallsetminus E))\\
        &\qquad - (W'(M_1 \cap E) + W'(M_1 \smallsetminus E)) \\
        & = (W'(M_2 \cap E) - W'(M_1 \cap E)) \\
        &\qquad + (W'(M_2 \smallsetminus E) - W'(M_1 \smallsetminus E)) \\
        & \geq 10 n^4 (W(M_2 \cap E) - W(M_1 \cap E)) - W'(M_1 \smallsetminus E) \\
        &\geq 10 n^4 - 10 n^3 > 0
    \end{align*}
    and thus $W'(M_1) < W'(M_2)$. Therefore the minimum weight perfect matching
    $M'$ of $G'$ must be an extension of $M^k_{\iso}$.  \\

    \noindent
    We say an extension is a minimum weight extension if $W'(M')$ is minimum amongst all
    extensions of $M$.
    We claim that the minimum weight extension of any size $k$ matching
    $M^k$ of $G$ is unique. \\

    \noindent
    Let $L_{\text{unmatched}}$ be vertices in $L$ and $R_{\text{unmatched}}$ be the vertices in
    $R$ which are not matched by $M^k$, and let $L_{\text{new}}$ and $R_{\text{new}}$ be the vertices
    in $L' \smallsetminus L$ and $R' \smallsetminus R$ respectively.
    Any extension of $M^k$ consists of a perfect matching in
    $(L_{\text{new}} \cup R_{\text{unmatched}}, E' \smallsetminus E)$ and in
    $(L_{\text{unmatched}} \cup R_{\text{new}}, E' \smallsetminus E)$, both of
    which are disjoint bipartite cliques whose weight function $W'$ is given by
    $w(u,v) = u \cdot v$. \\

    \noindent    
    We claim that both of these cliques have a unique
    perfect matching, and thus the min-weight extension of $M^k$ is unique:
    \begin{claim}
    \label{clique-isolation}
        Let $G = (V = L \cup R, E)$ be an $s \times s$ bipartite clique, and        
        let $W:E \rightarrow \mathbb{Z}$ be weights defined as $W(e = (u, v)) = u \cdot v$
        for an arbitrary indexing of $V$.
        Then $W$ isolates a perfect matching in $G$.
    \end{claim}
    \begin{proof}
        Let $L = \{l_1, \dots, l_s\}$ and $R = \{r_1, \dots, r_s\}$ be ordered by increasing indices. We claim that $\{(l_i, r_{n+1-i})\}_{i \in [n]}$ is the unique minimum weight perfect matching in $G$. For any other perfect matching $M$, there must exist $u_1 < u_2 \in L$ and $v_1 < v_2 \in R$ (again, ordered by indices) such that $(u_1, v_1), (u_2, v_2) \subseteq M$. Then for $M' = (M \smallsetminus \{(u_1, v_1), (u_2, v_2)\}) \cup \{(u_1, v_2), (u_2, v_1)\}$, it is easy to verify that $W(M') = W(M) -  (u_2 - u_1)(v_2 - v_1)$, and thus $W(M') < W(M)$, which shows that $M$ is not a minimum weight perfect matching.
    \end{proof}

    \noindent
    Thus $W$ isolates a perfect matching $M^*$ in $G'$, namely the unique minimum
    weight extension of the minimum weight perfect matching $M_{\iso}^k$ of size $k$.
    Applying Theorem~\ref{thm:isolated-pm} and Lemma~\ref{lem:DET-in-cl} shows
    that extracting $M^*$, and from it $M_{\iso}^k$, can be done in $\CLP$.
    %
    %
    %
\end{proof}




\section{Maximum Matching via Isolation}
\label{sec:induction}

In the previous section we showed how to extract a matching of
size $k$ in $G$; however, this requires us to know that $W$ isolates
such a matching for this value of $k$. In this section we show how
to find the maximum $k$ for which this holds.\footnote{An earlier result
of Hoang, Mahajan, and Thierauf~\cite{HoangMahajanThierauf06}
shows how to determine whether edge weights isolate a perfect matching,
assuming that one exists; however,
our algorithm needs to additionally distinguish between the case
where a size $k$ matching simply does not exist, and the case where a
size $k$ matching exists but is not isolated.} \\

\noindent
Our algorithm will do this for $k$ inductively.
This is easy to do for $k = 1$, as there are two min-weight matchings
iff there are two min-weight edges, and there are no matchings iff
there are no edges. Thus for the remainder of this section,
assume that $W$ isolates a unique min-weight matching $M_{\iso}^k$
of size $k$, computable in $\CLP$ by Lemma~\ref{lem:isolated-km},
which we will use to test $k+1$.

\subsection{Residual Graphs and Augmenting Paths}

The important tool in our construction will be the \textit{residual graph}
of our matching $M_{\iso}^k$. This construction is a standard tool in graph algorithms, and can be found, along with the other facts we state, in the monograph \cite{KarpinskiRytter98} or textbook \cite{CLRS}. For readers familiar with matching theory, this is exactly the residual graph one would obtain from the max flow instance corresponding to maximum matching.

\begin{definition}[Residual Graph]
    Let $G = (V = L \cup R, E)$ be a bipartite graph, $M$ be a matching in $G$,
    and $W:E \rightarrow \mathbb{Z}$ be edge weights.
    The \emphdef{residual graph} of $M$, which we denote by
    $G^M_{\res} = (V^M_{\res}, E^M_{\res})$, is a directed graph whose
    vertex set $V^M_{\res}$ consists of $V$ plus a new source $s$ and a new sink $t$.
    The edge set $E^M_{\res}$ consists of the following edges where $u \in L$
    and $v \in R$:
    \begin{enumerate}
        \item if $u$ is unmatched by $M$ we add an edge $(s \rightarrow u)$, and
        if $v$ is unmatched by $M$ we add an edge $(v \rightarrow t)$
        \item if $(u,v) \in M$ we add an edge $(v \rightarrow u)$
        \item if $(u,v) \in E \smallsetminus M$ we add an edge $(u \rightarrow v)$
    \end{enumerate}
    \noindent
    We will drop the $M$ superscript in cases where the usage is clear.
\end{definition}

\noindent
Because $M_{\iso}^k$ can be constructed in $\CLP$, the graph $G_{\res}$
can be constructed in $\CLP$ as well.
As we will soon see, paths from $s$ to $t$ in $G^M_{\res}$ (ignoring the first
and last edges) are closely related to matchings in $G$.

\begin{definition}[Augmenting Path]
Let $G = (V, E)$ be a graph, $M$ be a matching in $G$, and
$W:E \rightarrow \mathbb{Z}$ be a weight assignment.
A path $P = \{e_1, \dots, e_m\}$ on vertices $\{v_1, \dots, v_{m+1}\}$
is an \emphdef{augmenting path} with respect to $M$ if:
\begin{enumerate}
    \item $v_1$ and $v_{m+1}$ are not matched by $M$.
    \item For all odd $i \in [m]$, $e_i \notin M$.
    \item For all even $i \in [m]$, $e_i \in M$.
\end{enumerate}
\end{definition}

\noindent
We will be using a slightly different weight scheme for residual
graphs and augmenting paths, which will be convenient for talking
about matchings in their context.

\begin{definition}
Let $G = (V = L \cup R, E)$ be a bipartite graph, $M$ be a matching in $G$,
and $W:E \rightarrow \mathbb{Z}$ be a weight assignment.
We define the \emphdef{alternating} weight function as
\[W_{\alt}^M(e) \vcentcolon= 
     \begin{cases}
       -W(e), &\quad\text{if } e \in M\\
       W(e), &\quad\text{if } e \notin M \\
     \end{cases}
\]
Furthermore, for residual graph $G^M_{\res}$ we
define the \emphdef{residual} weight function as
\[W_{\res}^M(e) \vcentcolon= 
     \begin{cases}
       0, &\quad\text{if } e \notin E\\
       W_{\alt}^M(e), &\quad\text{otherwise} \\
     \end{cases}
\]
\noindent
Whenever we use weights in $G^M_{\res}$, we are referring to the weights $W^M_{\res}$,
while for augmenting paths $P$ with respect to $M$ we use $W_{\alt}^M$.
\end{definition}

\noindent
It is straightforward to observe that any augmenting path
$P$ in $G$ can be extended to a path in the residual graph
$G_{\res}^M$ via two additional edges and vice versa, and
furthermore both paths have the same weight; we will analyze
this fact in more detail later. \\

\noindent
Important to our algorithm is that augmenting paths to size $k$ matchings in $G$ are, in a sense,
equivalent to matchings of size $k+1$.

\begin{fact}[Berge's Theorem~\cite{Berge57}]
\label{fact:augmented-matching}
    Let $G = (V, E)$ be a graph, $M$ be a matching in $G$,
    and $W:E \rightarrow \mathbb{Z}$ be a weight assignment.
    Then $M$ is maximum iff there do not exist any augmenting paths with respect to $M$.
    Furthermore, for any augmenting path $P$, $M \oplus P$ is a matching in $G$
    of size $|M| + 1$.  
\end{fact}

\noindent
An immediate consequence of Fact~\ref{fact:augmented-matching} is that
we can test if there are no matchings of size $k+1$.

\begin{lemma}
\label{lem:check-if-matching-is-maximum}
    Let $G = (V = L \cup R, E)$ be a bipartite graph and $W:E \rightarrow \mathbb{Z}$
    be a weight assignment.
    Let $k \in [n]$ such that $W$ isolates a size $k$ matching $M^k_{\iso}$.
    Then there exists a $\CLP$ machine which, given $(G, W, k)$, decides whether
    $M^k_{\iso}$ is a maximum matching.
\end{lemma}
\begin{proof}
    By Lemma~\ref{lem:isolated-km} we can compute $M^k_{\iso}$, and from it we
    can build the residual graph $G_{\res}$. 
    By Fact~\ref{fact:augmented-matching}, $M^k_{\iso}$ is a maximum matching
    iff $t$ is not reachable from $s$ in $G_{\res}$, which we can check in $\CLP$.
\end{proof}

\noindent
A sequential algorithm could of course repeatedly apply \Cref{lem:check-if-matching-is-maximum} to search for larger matchings, as proposed in \cite{Kuhn55}. However, this requires large space to store the matching after a certain number of iterations. Instead, we will utilise a hybrid approach between the isolation lemma based approach of \cite{MulmuleyVaziraniVazirani87} and the augmenting paths based approach of \cite{Kuhn55}.


\subsection{Matching Weights and Augmenting Paths}

For the rest of this section we assume that at least one augmenting path
exists, and thus there is some matching of size $k+1$. Our goal is to
use $G_{\res}$ to understand all matchings of size $k+1$,
and in particular the min-weight matchings.   

\begin{lemma}
\label{lem:sym-dif}
    Let $G = (V = L \cup R, E)$ be a bipartite graph and $W:E \rightarrow \mathbb{Z}$
    be edge weights such that $W$ isolates a size $k$ matching $M^k_{\iso}$ in $G$.
    Let $M^{k+1}$ be any minimum weight size $k+1$ matching in $G$.
    Then the symmetric difference $M^k_{\iso} \Delta M^{k+1}$ is a single
    augmenting path with respect to $M^k_{\iso}$.
\end{lemma}
\begin{proof}
    Define $H := M^k_{\iso} \Delta M^{k+1}$, and assume for contradiction that $H$ is
    not a single augmenting path. We will go through the different potential cases for $H$,
    showing that either $M^k_{\iso}$ or $M^{k+1}$ can be modified to contradict either
    their minimality or, in the case of $M^k_{\iso}$, its uniqueness. \\

    \noindent    
    The proof of the following claim is in the spirit of \cite{DattaKulkarniRoy10}, which introduced the relationship between alternating weights and isolated matchings.
    \begin{claim}
    \label{claim:sym-dif-helper}
        There does not exist any set $\emptyset \subsetneq S \subseteq H$ satisfying the following properties:
        \begin{enumerate}
            \item $M^k_{\iso} \oplus S$ is a matching in $G$ of size $k$.
            \item $M^{k+1} \oplus S$ is a matching in $G$ of size $k+1$.
        \end{enumerate}
    \end{claim}
    \begin{proof}
        Note that, since every $e \in S$ belongs to either $M^k_{\iso}$ or $M^{k+1}$ and
        not the other, we have
        \[W^{M^{k}_{\iso}}_{\alt}(S) = \sum_{e \in S \cap M^k_{\iso}} -W(e) + \sum_{e \in S \cap M^{k+1}}W(e) = -W^{M^{k+1}}_{\alt}(S)\]
        \noindent
        There are two cases:
        \begin{enumerate}
            \item $W^{M^{k}_{\iso}}_{\alt}(S) \leq 0$ : In this case, $M^k_{\iso} \oplus S$
            is a distinct matching of size $k$ with weight
            $W(M^k_{\iso}) + W^{M^{k}_{\iso}}_{\alt}(S) \leq W(M^k_{\iso})$,
            which contradicts the fact that $M^k_{\iso}$ is the \textit{unique}
            minimum weight matching of size $k$.
            \item $W^{M^{k}_{\iso}}_{\alt}(S) > 0$: In this case, $M^{k+1} \oplus S$
            is a matching of size $k+1$ with weight
            $W(M^{k+1}) + W^{M^{k+1}}_{\alt}(S) = W(M^{k+1}) - W^{M^{k}_{\iso}}_{\alt}(S) < W(M^{k+1})$, which contradicts the minimality of $M^{k+1}$. \qedhere
        \end{enumerate}
    \end{proof}

    \noindent
    We now show that such an $S$ must exist in any $H$ which does not consist of a single augmenting
    path. Clearly any path in $H$ alternates between edges of $M^k_{\iso}$ and $M^{k+1}$,
    and thus the connected components of $H$ are of the following four types:
    \begin{enumerate}
        \item even length alternating cycles.
        \item even length alternating paths.
        \item augmenting paths with respect to $M^k_{\iso}$. 
        \item augmenting paths with respect to $M^{k+1}$.  
    \end{enumerate}
    The former two cases are immediate; any even length cycle or path has an equal number
    of edges in both $M_{\iso}^k$ and $M^{k+1}$, and so taking $S$ to be this
    contradicts Claim~\ref{claim:sym-dif-helper}.
    Thus we can assume all connected components of $H$ are augmenting paths with respect to
    $M^k_{\iso}$ or $M^{k+1}$. \\

    \noindent
    Let $P_{k}$ be the set of augmenting paths in $H$ with respect to
    $M^k_{\iso}$, and let $P_{k+1}$ be the set of augmenting paths with respect to $M^{k+1}$.
    Since $|M^{k+1}| = |M^k_{\iso}| + 1$, we have $|P_{k+1}| = |P_k| + 1$, and by assumption
    we do not have $|P_{k+1}| = 1$ and $|P_k| = 0$; thus $P_k \neq \emptyset$.
    Define $S = P \cup P'$ for any any $P \in P_k$ and $P' \in P_{k+1}$, which
    gives us a contradiction by applying Claim~\ref{claim:sym-dif-helper}.
\end{proof}

\noindent
Because paths in $G_{\res}$ correspond to augmenting paths of equal weight,
Lemma~\ref{lem:sym-dif} implies that isolating $k+1$-size matchings is
equivalent to isolating paths in $G_{\res}$.

\begin{lemma}
\label{lem:isolation-equivalence}
    Let $G = (V = L \cup R, E)$ be a bipartite graph and $W:E \rightarrow \mathbb{Z}$ be
    edge weights such that $W$ isolates a size $k$ matching $M^k_{\iso}$ in $G$.
    Then $W$ isolates a size $k+1$ matching in $G$ iff
    $W_{\res}$ isolates an $s$-$t$ path in $G_{\res}$.
\end{lemma}
\begin{proof}
    We show that both are equivalent to showing $W_{\alt}$ isolates an augmenting path
    with respect to $M^k_{\iso}$.
    Let $M^{k+1}$ be a minimum weight size $k+1$ matching in $G$.
    By Lemma~\ref{lem:sym-dif}, $M^{k+1} \Delta M^k_{\iso}$ is an augmenting path
    with respect to $M^k_{\iso}$ of weight
    \[W(M^{k+1}) = W(M^k_{\iso}) + W_{\alt}(M^{k+1} \Delta M^k_{\iso})\]
    \noindent
    Moreover, for any augmenting path $P$ with respect to $M^k_{\iso}$,
    \[W(M^k_{\iso} \oplus P) = W(M^k_{\iso}) + W_{\alt}(P)\]
    \noindent
    By the minimality of $M^{k+1}$, \[W_{\alt}(M^{k+1} \Delta M^k_{\iso}) \leq W_{\alt}(P)\]
    for every augmenting path $P$ with respect to $M^k_{\iso}$.
    Thus, $M^{k+1} \Delta M^k_{\iso}$ is a minimum weight augmenting path,
    and for any minimum weight augmenting path $P$, $M^k_{\iso} \oplus P$
    is a minimum weight size $k+1$ matching. \\

    \noindent    
    If $W$ does not isolate a size $k+1$ matching, there exist at least two
    minimum weight matchings of size $k+1$, and we let
    $M^{k+1}_1$ and $M^{k+1}_2$ be any two such matchings.
    Since they are distinct matchings, we conclude that $M^{k+1}_1 \Delta M^k_{\iso}$
    and $M^{k+1}_2 \Delta M^k_{\iso}$ are distinct minimum weight augmenting paths
    with respect to $M^k_{\iso}$.
    %
    Conversely, assume that there exist distinct minimum weight augmenting paths $P_1$
    and $P_2$. Then $M^k_{\iso} \oplus P_1$ and $M^k_{\iso} \oplus P_2$ are
    distinct minimum weight size $k+1$ matchings. \\

    \noindent
    Second, we show that $W_{\alt}$ isolates an augmenting path with
    respect to $M^k_{\iso}$ iff $W$ isolates an $s$-$t$ path in $G_{\res}$.
    As observed before, every augmenting path with respect to $M^k_{\iso}$ gives
    a simple $s$-$t$ path in $G_{\res}$ and vice versa.
    It is thus sufficient to show that the shortest $s$-$t$ path in $G_{\res}$
    is always simple, i.e. $G_{\res}$ does not contain any cycle $C$ such that
    $W_{\res}(C) \leq 0$ with respect to $W$.
    Assuming otherwise, since $C$ contains an equal number of edges
    in and not in $M^k_{\iso}$, taking $M^k_{\iso} \oplus C$ gives
    us a new matching of size $k$ and weight
    $$W(M^k_{\iso}) + W_{\res}(C) = W(M^k_{\iso}) + W_{\alt}(C) \leq W(M^k_{\iso})$$
    which contradicts the unique minimality of $M^k_{\iso}$.
    %
\end{proof}

\section{Compression of Weight Assignments}
\label{sec:compress}

In this section, we combine our previous algorithms with a new
step for $\CLP$, namely compressing the weight function in the
case when it does not isolate a matching of some size $k$.
This will allow us to use our catalytic tape itself as providing
a weight function, since it will either a) be random enough
to successfully run the algorithm as outlined above,
or b) we will be able to compress enough space to compute $\MATCH$ in $\P$.

\subsection{Recursive Step: Termination, Isolation, or Failure}

We begin by identifying the information
that will be needed in the case when
$W$ does not isolate a matching of size $k+1$.

\begin{definition}[Threshold Edge]
    Let $G = (V, E)$ be a directed graph, $s \in V$ be a source vertex,
    $t \in V$ be a target vertex,
    and $W: E \rightarrow \mathbb{Z}$ be edge weights such that $G$ does
    not have any non-positive weight cycles under $W$.
    We say that an edge $e \in E$ is a \emphdef{threshold edge}
    if there exist two minimum weight $s$-$t$ paths $P_1, P_2$ in $G$
    such that $e \in P_1$ and $e \notin P_2$.
    %
\end{definition}

\noindent
By definition it is clear that a threshold edge exists iff $W$ does
not isolate an $s$-$t$ path in $G$. We also observe
that by computing reachability in $\CLP$, we can find such an edge:





\begin{lemma}
\label{lem:compute-threshold-edge}
    Let $G = (V, E)$ be a directed graph, $s \in V$ be a source vertex,
    $t \in V$ be a sink vertex, and $W: E \rightarrow \mathbb{Z}_{\leq \poly(n)}$
    be edge weights such that $G$ does not have any non-positive weight cycles under $W$.
    Furthermore, assume that there exist at least two distinct min-weight $s$-$t$ paths
    under $W$.
    Then there exists a $\CLP$ machine which, given $(G, s, t, W)$, outputs a threshold edge.
\end{lemma}
\begin{proof}
    For any two vertices $u,v \in G$, let $\delta_{u,v}$ be the minimum weight of any
    $u$-$v$ path. For each edge $e = u \rightarrow v \in E$, $e$ is a threshold
    edge iff 1) $\delta_{s,u} + W(e) + \delta_{v,t} = \delta_{s,t}$, i.e. $e$ lies
    on at least one min-weight $s$-$t$ path; and
    2) $\delta_{s,t}' = \delta_{s,t}$, where $\delta_{s,t}'$ is the min-weight
    $s$-$t$ path in $G \smallsetminus \{e\}$, i.e. there exists some min-weight
    $s$-$t$ path which does not use $e$.
    All these tests can compute in $\CLP$ using Lemma~\ref{lem:compute-min-weight-path},
    and so we loop over all $e$ in order until we find a threshold edge.
    %
    %
    %
\end{proof}

\noindent
This finally brings us to our recursive procedure, which allows us to
go from isolating matchings of size $k$ to matchings of size $k+1$.

\begin{lemma}
\label{lem:check-whether-k+1-is-isolated}
    Let $G = (V = L \cup R, E)$ be a bipartite graph, let $W:E \rightarrow \mathbb{Z}$ be
    a weight assignment, and let $k \in [n]$ such that $W$ isolates a size
    $k$ matching $M^k_{\iso}$.
    Then there exists a $\CLP$ machine $A$ which, on input $(G,k,W)$,
    performs the following:
    \begin{enumerate}
        \item if no matching of size $k+1$ exists, $A$ outputs $\perp$
        \item if $W$ isolates a matching of size $k+1$, $A$ outputs $1$
        \item if $W$ does not isolate a matching of size $k+1$, $A$ outputs a threshold
        edge $e$ in $G_{\res}$ with the further promise that $e \notin M^k_{\iso}$
    \end{enumerate}
    %
\end{lemma}
\begin{proof}
    In $\CLP$ we can compute $M_{\iso}^k$ by Lemma~\ref{lem:isolated-km}, and thus
    we can construct the graph $G_{\res}$.
    By Lemma~\ref{lem:check-if-matching-is-maximum}, if no matching of size $k+1$
    exists, we can determine this in $\CLP$.
    Otherwise, we run the algorithm of Lemma~\ref{lem:compute-threshold-edge}
    to see if any threshold edge exists, and if not then $W$ must isolate
    a matching of size $k+1$ by Lemma~\ref{lem:isolation-equivalence}. \\

    \noindent    
    Finally, if any threshold edge exists, then there must exist one
    outside $M^k_{\iso}$. If not, then every min-weight path
    contains the same set of edges outside $M^k_{\iso}$,
    and since each vertex is only adjacent to at most one edge in
    $M^k_{\iso}$ every min-weight path must also
    contain the same set of edges within
    $M^k_{\iso}$. Thus every min-weight matching of size $k+1$
    is the same, which is a contradiction.
\end{proof}


\subsection{Compressing Isolating Edges}

Since it is clear how to proceed in the first two cases of
Lemma~\ref{lem:check-whether-k+1-is-isolated},
we now turn to the third case, when we
obtain a threshold edge. The key observation is that $W(e)$
can be determined via the rest of the weight function.

\begin{lemma}
\label{lem:decompression-threshold-edge}
    Let $G = (V, E)$ be a directed graph, $s \in V$ be a source vertex and $t \in V$
    be a target vertex, and $W: E \rightarrow \mathbb{Z}_{\leq \poly(n)}$
    be edge weights such that all cycles in $G$ have strictly positive weight under $W$.
    Let $e = u \rightarrow v$ be a threshold edge in this graph.
    Then there exists a $\CLP$ machine which, given
    $(G, e, W \restriction_{E \smallsetminus \{e\}})$, computes $W(e)$.
\end{lemma}
\begin{proof}
    This algorithm is essentially the same as Lemma~\ref{lem:compute-threshold-edge}.
    To recap, we again define, for any two vertices $u,v \in G$, the value
    $\delta_{u,v}$ be the minimum weight of any
    $u$-$v$ path. Since $e = u \rightarrow v \in E$ is a threshold
    edge, it follows that $\delta_{s,t} = \delta_{s,u} + W(e) + \delta_{v,t}$,
    and furthermore that $\delta_{s,t} = \delta_{s,t}'$ where $\delta_{s,t}'$ is
    the min-weight $s$-$t$ path in $G \smallsetminus \{e\}$. Putting these two
    facts together we get that
    $$W(e) = \delta_{s,t}' - (\delta_{s,u} + \delta_{v,t})$$
    and we can calculate all three quantities in $\CLP$ using
    Lemma~\ref{lem:compute-min-weight-path}. Furthermore, none of
    these quantities involve $W(e)$; this is true for $\delta_{s,t}'$
    by definition, while the other two hold because all cycles have
    positive weight and thus every min-weight path from $s$ to $u$
    (or $v$ to $t$) does not involve $u \rightarrow v$.
\end{proof}

\noindent
This gives rise to our compression and decompression subroutines, which allow
us to erase and later recover $W(e)$ from the catalytic tape.

\begin{lemma}
\label{lem:decompression}
    Let $G = (V = L \cup R, E)$ be a bipartite graph, and let $\tau$ be a string
    of length $\poly(n)$ such that we interpret the initial substring of $\tau$
    as a weight assignment $W:E \rightarrow \mathbb{Z}_{\leq n^5}$ which isolates a matching
    of size $k$ but does not isolate a matching of size $k+1$. We interpret an additional $5 \log n$ bits on the catalytic tape as a reserve weight $r \in \mathbb{Z}_{\leq n^5}$. \\

    \noindent
    Then there exist catalytic subroutines $\Comp, \Decomp$ with the following behavior:
    \begin{itemize}
        \item $\Comp(G,k)$ replaces $(W,r)$ with $(W', k, e, 0^{2 \log n})$ for some edge $e$, where $W'(e') = W(e')$ for all $e' \neq e$ and $W'(e) = r$, and leaves all other catalytic memory unchanged
        \item $\Decomp(G,k,u,v)$ inverts $\Comp$
    \end{itemize}    
\end{lemma}
\begin{proof}
    By Lemma~\ref{lem:check-whether-k+1-is-isolated}, $\Comp$ can use the
    rest of its catalytic tape to find a threshold edge $e = u \rightarrow v$
    which is outside of $G_{\iso}^k$, given that such
    an edge must exist by assumption. Now $\Comp$ will erase $W(e)$ from the catalytic
    tape and replace it by $r$; we then erase the original copy of $r$ and record
    $k$ and the indices $u, v$ of $e$ in it.
    These indices take $\log n$ bits each and $k \leq n$
    requires $\log n$ bits, while $W(e)$ takes $5 \log n$ bits on the catalytic tape,
    which gives us $2 \log n$ free bits as a result of this procedure. \\

    \noindent    
    For $\Decomp$, since $e \notin M_{\iso}^k$, we can determine $M_{\iso}^k$
    by Lemma~\ref{lem:isolated-km} given $(G \smallsetminus \{e\},k,W \smallsetminus \{W(e)\})$,
    and from this we can construct $G_{\res}$;
    thus we can apply Lemma~\ref{lem:decompression-threshold-edge}
    to recover the value $W(e)$. We then erase $k$ and $e$, move the weight in
    $W'$ at location $e$ to this memory, and replace it with the recovered
    value of $W(e)$.
    %
\end{proof}

\section{Final Algorithm}
\label{sec:final-alg}

We finally collect all cases together to solve bipartite maximum matching
in $\CLP$ and in $\LOSSYNC$.

\subsection{Proof of Theorem~\ref{thm:main}}

First we show the case of $\CLP$, where our core algorithm will
be spelled out in detail.

\begin{theorem} \label{thm:main-specific}
    There exists a $\CLP$ algorithm which, given bipartite graph $G$
    as input, outputs the maximum matching in $G$.
\end{theorem}

\begin{proof}
    By \cite{HopcroftKarp}, $\MATCH$ can be solved in
    time $T := O(|E| \sqrt{|V|})$.
    Our $\CLP$ algorithm will have three sections of catalytic tape:
    \begin{enumerate}
        \item a weight assignment $W: E \rightarrow \mathbb{Z}_{\leq n^5}$
        \item a set of reserve weights $r_1 \ldots r_{T/(2 \log n)}$, each of size $5 \log n$
        \item a large enough $\poly(n)$ catalytic space to run all $\CLP$ subroutines
        as needed
    \end{enumerate}    
    We set two loop counters $c$ and $k$, both initialized to 0; $c$ will
    record how many times we have compressed a weight value, while $k$ will
    record the largest isolated min-weight matching found by $W$ in the current
    iteration. Our basic loop for the current $c$ and $k$ will be to
    apply Lemma~\ref{lem:check-whether-k+1-is-isolated} for the current $(G,k,W)$
    and perform as follows:
    \begin{enumerate}
        \item if it returns $\perp$, record $M_{\iso}^k$ on the output tape
        and move to the decompression procedure (see below).
        \item if it returns 1, increment $k$ and repeat.
        \item if it returns an edge $e$, increment $c$, apply the $\Comp$ subroutine of
        Lemma~\ref{lem:decompression} using reserve weight $r_c$ as $r$,
        and restart our algorithm for $k = 0$ with our new weight function $W'$
    \end{enumerate}
    If our algorithm ever reaches the first case, we have successfully computed
    the maximum matching in $G$. This occurs unless we reach $c = T/(2 \log n)$,
    and in this case we are left with $T$ free bits on our catalytic tape,
    as each application of $\Comp$ frees $2 \log n$ bits.
    We then apply our time $T$ algorithm to solve $\MATCH$ directly.
    We record our answer on the output tape and move to the decompression procedure. \\

    \noindent
    To decompress the tape, we apply the $\Decomp$ procedure of
    Lemma~\ref{lem:decompression} in reverse order, starting from our final
    value of $c$ and decrementing until we reach 0.
    By the correctness of $\Decomp$ each iteration will reset the catalytic
    tape to its state just before the $c$th run of the algorithm, meaning
    that our final state is the original catalytic tape $\tau$, at which point
    we return the answer saved on our work tape.\\

    \noindent
    We briefly analyze our resource usage. Our catalytic tape has length
    $$ \left({n\choose 2} + T/(2 \log n) \right) \cdot (5 \log n) + \poly(n) = \poly(n) $$
    Our work tape will need to store loop variables $k \leq n$ and $c \leq T/(2 \log n)$,
    plus free space to run our $\CLP$ subroutines, which is $O(\log n)$ in total.\\

    \noindent
    All subroutines are $\CLP$ machines and thus run in polynomial time,
    while our loops for $k$ and $c$ run in time $n$ and $T/(2 \log n)$ respectively,
    and the decompression procedure again only uses $\CLP$ subroutines and runs
    for $c$ steps.
    Finally if we reach the maximum value of $c$, we ultimately run the
    time $T$ algorithm, which again take polynomial time. Thus our whole machine
    runs in polynomial time, logarithmic free space, and polynomial catalytic space,
    which is altogether a $\CLP$ algorithm.
\end{proof}

\begin{remark}
    Putting aside our runtime analysis and care with regards to using $\CLP$ rather
    than $\CL$ subroutines,
    it is also known due to Cook et al.~\cite{CookLiMertzPyne25}---in
    fact, by the same argument structure that we use here---that any problem in
    $\CL \cap \P$ can be solved in poly-time bounded $\CL$ generically.
\end{remark}

\subsection{Proof of Theorem~\ref{thm:lossy}}
We now prove our second theorem using the above algorithm; since
most of the details are analogous we opt to be somewhat succinct
in our proof.

\begin{theorem} \label{thm:lossy-specific}
There exist $\NC$ algorithms $\mathcal{A}_1$, $\mathcal{A}_2$ which, given a bipartite graph $G = (V, E)$ with as input, have the following behaviour:
\begin{enumerate}
    \item $\mathcal{A}_1$ outputs $\NC$ circuits $\Comp:\{0, 1\}^{f(n)} \rightarrow \{0, 1\}^{f(n) - 1}$ and $\Decomp:\{0, 1\}^{f(n) - 1} \rightarrow \{0, 1\}^{f(n)}$ such that both $\Comp$ and $\Decomp$ have depth $\polylog(n)$ and size $\poly(n)$. Furthermore, $f(n) \in \poly(n)$. 
    \item Given any $x \in \{0, 1\}^{f(n)}$ such that $\Decomp(\Comp(x)) \neq x$,
    $\mathcal{A}_2$ outputs a maximum matching of $G$.
\end{enumerate}
\end{theorem}

\begin{proof}
Our proof follows from Theorem~\ref{thm:main-specific}, and in particular
Lemma~\ref{lem:decompression}, under a different setup.
In particular, rather than using $\CLP$ subroutines and a weight function from the
catalytic tape, we will use $\NC$ subroutines and a weight function as given
by the input. Our outer $c$ loop is unnecessary as we only need to compress
once, while the inner $k$ loop can be parallelized to keep our algorithm in low depth.\\

\noindent
We now move to the details of how to construct our algorithms and circuits.
We first describe the $\NC$ algorithm which corresponds to $\Comp$:
\begin{enumerate}
    \item We interpret the input of circuit $C$ as the graph $G$ along with edge weights $W:E \rightarrow Z_{\leq \poly(n)}$ with $3\log n + 1$ bits per weight; thus $f(n) = |E| \cdot (3 \log n + 1) = \poly(n)$.
    \item For all $k \in [0, n]$ in parallel, use \Cref{lem:isolated-km} to compute a matching of size $k$, which succeeds if one is isolated by $W$. 
    \item For all $k \in [0, n]$ such that \Cref{lem:isolated-km} gives a size $k$ matching, use \Cref{lem:check-whether-k+1-is-isolated} to check whether a matching of size $k+1$ is isolated. 
    \item Let $k^*$ be the minimum $k$ for which \Cref{lem:isolated-km} returns a size $k$ matching of $G$, but \Cref{lem:check-whether-k+1-is-isolated} returns a threshold edge $e^*$. If no such $k^*$ exists, it implies that $W$ isolates a maximum matching of $G$. In this case $C$ is not required to compress accurately, so it can simply output the all $0$s string $0^{f(n) - 1}$.
    \item For all $0 \leq k \leq k^*$, we are guaranteed that $W$ isolates a size $k$ matching in $G$. We know that $W$ does not isolate a size $k^* + 1$ matching and, from the previous step, we have a threshold edge $e^*$.
    \item $\Comp$ now outputs the original weight function $W$ but with the $3\log n + 1$ bits corresponding to weight $W(e^*)$ replaced with $3\log n$ bits representing $k^*$ and $e^*$, which we move to the end of the output for simplicity.
\end{enumerate}

\noindent
Now we describe the $\NC$ algorithm which corresponds to $\Decomp$, and we only
consider the case that the compression succeeds, i.e. does not output $0^{f(n)-1}$,
as the other case is irrelevant:
\begin{enumerate}
    \item Read the last $3\log(n)$ bits of our input and interpret them as $k^*$ and $e^*$ as described above. 
    \item Using \Cref{lem:isolated-km}, construct the isolated size $k^*$ matching $M^{k^*}_{\iso}$ in the graph $G \setminus \{e^*\}$.
    \item Using $M^{k^*}_{\iso}$, construct its residual graph $G_{\res}$, and then run the procedure described in \Cref{lem:decompression-threshold-edge} to obtain $W(e^*)$.
    \item Erase the $3\log n$ bits in the suffix of the input and output the weight function given as input with the recomputed $3\log n + 1$ bit weight $W(e^*)$ in its appropriate position.
\end{enumerate}

\noindent
It is very easy to verify that all of the algorithms in the referenced lemmas work in $\NC$, as they only use subroutines from $\Logspace, \NL$, and $\TCo$, and operate in parallel for all $k$. Thus, both $\Comp$ and $\Decomp$ are in uniform $\NC$, and the algorithm $\mathcal{A}_1$ simply constructs the circuits $C$ and $D$ corresponding to the uniform $\NC$ algorithms above. \\

\noindent
For $\mathcal{A}_2$ we simply observe that the algorithm $\Comp$ gives us the matching
in the case when it fails to compress, namely in the case where no $k$ outputs a
threshold edge. Thus, given such a weight function, the algorithm $\mathcal{A}_2$
can simply run over all $k \in [0, n]$ in parallel and use \Cref{lem:isolated-km}
to attempt to construct a size $k$ matching.
The largest matching $M^k_{\iso}$ for which this algorithm succeeds is guaranteed
to be the maximum matching, which it can simply output.
Thus $\mathcal{A}_2$ is also an $\NC$ algorithm by the same argument as $\Comp$. 
\end{proof}

\subsection{Postscript: a note on the isolation lemma}

To close the main section of our paper, we note an interesting feature of our approach
as discussed in the introduction. Our key observation in Section~\ref{sec:compress}
is that in the case where weights $W$ are not isolating, an edge $e$ exists which is
in \textit{some} minimum weight matchings, \textit{but not all} of them.
This is, in fact, a general feature of non-isolating weights on arbitrary families of sets.
In the original proof of \cite{MulmuleyVaziraniVazirani87}, they refer to this element
as being ``on the threshold'' (hence our use of the term ``threshold edge''),
and use it to analyse the probability of failure.\\

\noindent
In theory, the weight of a threshold element could always be erased and reconstructed later.
Thus, the approach we have described here could be used to derandomize in $\CL$
(or in $\LOSSYC$) any algorithm which employs the isolation lemma.
The bottleneck is of course designing efficient compression-decompression routines
for these problems (which correspond to the circuits $\Comp$ and $\Decomp$ in the case
of lossy coding). Our contribution is thus twofold: we observe that the weight of
a threshold element can be erased and later reconstructed,
and we design a novel approach to identify a threshold element and reconstruct its
weight in the special case of $\MATCH$.




\section{Related Problems}
\label{sec:others}

In this section we discuss the implications of our result for related problems.

\begin{corollary} \label{cor:other-qs}
    The following search problems are in $\CLP$:
    \begin{enumerate}
        \item minimum weight maximum matching with polynomially bounded weights
        \item directed $s$-$t$ maximum flow in general graphs with polynomially bounded capacities
        \item global minimum cut in general graphs
    \end{enumerate}
\end{corollary}

\begin{proof}
    We sketch how each point follows from our earlier algorithm in turn.
    
    \paragraph{Min-weight matching.}
    Let $W_{\inp}$ be the edge weights given as input, with respect to which
    we want to find a minimum weight maximum matching. We will again iteratively
    use a weighting scheme $W_{\cata}$ as given on the catalytic tape
    by using edge weights $W := W_{\inp} \cdot n^{100} + W_{\cata}$.
    As with our original algorithm, in each iteration we either
    find a minimum weight maximum matching according to $W$,
    which is clearly also a minimum weight maximum matching for weights
    $W_{\inp}$, or we find a redundant weight in $W$, which also gives
    a redundant weight in $W_{\cata}$ which we can compress.
    As before we ultimately free up $\poly(n)$ space on the catalytic tape
    and again use to run any deterministic polytime algorithm for minimum
    weight maximum matching, such as the one in \cite{Kuhn55}.
    
    \paragraph{Directed max-flow.}
    This directly follows by known reductions: Madry \cite{Madry13}
    showed a logspace reduction from $\log$-bit directed $s$-$t$ maximum
    flow to $\log$-bit bipartite $b$-matching, and there is a trivial
    reduction from $\log$-bit bipartite $b$-matching to bipartite matching.
    
    \paragraph{Global min-cut.}
    We first note that this is immediate for $\log$-bit weights, as one can compute the
    $s$-$t$ minimum cut for every $(s, t) \in V \times V$ using the
    max-flow algorithm described above and simply take the minimum. \\
    
    We now move onto the case of $\poly(n)$-bit weights.
    Karger and Motwani~\cite{KargerMotwani94} proved that the following
    algorithm converts $\poly$-bit weights to $\log$-bit weights,
    such that the global minimum cut with respect to the original weights
    remains a $2$-approximate global minimum cut with respect to the new weights:
    \begin{enumerate}
        \item Construct a maximum spanning tree $T$. Let the minimum weight of
        any edge in this tree be $w$. 
        \item For edges $e$ such that $W(e) > n^2w$, set $W'(e) = n^{10}$.
        For all other edges $e$ set $W'(e) = W(e) \cdot n^3/w$, rounded off
        to the nearest integer.
    \end{enumerate}
    Using Reingold's celebrated result~\cite{Reingold08} that undirected
    $s$-$t$ connectivity is in log-space, one can construct a
    maximum spanning tree in log-space, since edge $e = (u, v)$ is in the
    maximum spanning tree iff $u$ is not reachable from $v$ using
    edges of weight $\geq W(e)$ (for ties, we additionally filter to
    edges $e'$ with greater index $e' > e$). Using this fact, the
    aforementioned reduction is in $\CLP$. \\

    \noindent    
    Thus, we simply need to enumerate the $2$-approximate global minimum cuts of
    $G$ with respect to a $\log$-bit weight function. A recent result of
    Beideman, Chandrawekaran, and Wang~\cite{BeidemanChandrasekaranWang23} shows
    that for every such cut $(C, \bar{C})$, there exist sets $S \subseteq C$, $T \subseteq \bar{C}$ with $|S|, |T| \leq 10$ such that $(C, \bar{C})$ is the unique minimum cut separating $S$ from $T$. Thus, we can simply iterate over all such sets $S$ and $T$,
    contract $S$ and $T$ into single vertices, and compute the minimum $S$-$T$
    cut again using our $\CLP$ max-flow algorithm and the max-flow/min-cut theorem.
    Finally, we output the cut that is minimal with respect to the original weights. 
\end{proof}

\ifblind
\section*{Acknowledgements}
The first author thanks Danupon Nanongkai and Samir Datta for lengthy and
insightful discussions about bipartite matching and the isolation lemma.
The second author thanks Michal Kouck\'{y}, Ted Pyne, Sasha Sami, and Ninad Rajgopal
for early conversations about compression and the isolation lemma.
Both authors thank Samir Datta, Danupon Nanongkai, and Ted Pyne
for detailed comments on an earlier draft, as well as Ted Pyne and Roei Tell
for discussions on lossy coding.
\fi


\DeclareUrlCommand{\Doi}{\urlstyle{sf}}
\renewcommand{\path}[1]{\small\Doi{#1}}
\renewcommand{\url}[1]{\href{#1}{\small\Doi{#1}}}
\bibliographystyle{alphaurl}
\bibliography{bibliography}






\end{document}